\documentclass{llncs}

\usepackage{amsmath}
\usepackage[utf8]{inputenc}
\usepackage[T1]{fontenc}
\usepackage{amsfonts}
\usepackage{amssymb}
\usepackage{float}
\usepackage{stmaryrd}
\usepackage{url}
\usepackage{hyperref}
\usepackage{enumitem}
\usepackage[margin=2cm]{geometry}
\usepackage{tikz}
  \usetikzlibrary{automata}
  \usetikzlibrary{arrows}
  \usetikzlibrary{shapes}
  \usetikzlibrary{decorations.pathmorphing}
  \usetikzlibrary{fit}
  \usetikzlibrary{trees}
  \usetikzlibrary{calc}

\tikzstyle{every picture}=[
  >=stealth', shorten >=1pt, node distance=1.44cm,auto,bend angle=45,initial text=,
  every state/.style={inner sep=0.75mm, minimum size=1mm},font=\scriptsize,
  triangle/.style = { regular polygon, regular polygon sides=3, inner sep=0}
]

\allowdisplaybreaks[1]

\begin{document}

\title{The Bottom-Up Position Tree Automaton, the Father Automaton and their Compact Versions}

\author{Samira Attou\inst{1} \and Ludovic Mignot\inst{2} \and Djelloul Ziadi\inst{2}}

\institute{
   USTHB, Faculty of Mathematics, RECITS Laboratory,\\
   BP 32, El Alia, 16111 Bab Ezzouar, Algiers, Algeria
  \and
    Groupe de Recherche Rouennais en Informatique Fondamentale,\\
    Université de Rouen Normandie,\\
    Avenue de l'Université,
    76801 Saint-Étienne-du-Rouvray, France
  \email{sattou@usthb.dz,\{ludovic.mignot,djelloul.ziadi\}@univ-rouen.fr}
}

\maketitle

\begin{abstract}
  The conversion of a given regular tree expression into a tree automaton has been widely studied.
  However, classical interpretations are based upon a Top-Down interpretation of tree automata.
  In this paper, we propose new constructions based on the Gluskov's one and on the one of Ilie and Yu one using a Bottom-Up interpretation.
  One of the main goals of this technique is to consider as a next step the links with deterministic recognizers, consideration that cannot be performed with classical Top-Down approaches.
  Furthermore, we exhibit a method to factorize transitions of tree automata and show that this technique is particularly interesting for these constructions, by considering natural factorizations due to the structure of regular expression.
\end{abstract}

\section{Introduction}\label{se:int}
Automata are recognizers used in various domains of applications  especially in computer science, \emph{e.g.} to represent (non necessarily finite) languages, or to solve the membership test, \emph{i.e.} to verify whether a given element belongs to a language or not.
Regular expressions are compact representations for these recognizers.
Indeed, in the case where elements are words, it is well known that each regular expression can be transformed into a finite state machine recognizing the language it defines.
Several methods have been proposed to realize this conversion.
As an example, Glushkov~\cite{Glu61} (and independently Mc-Naughton and Yamada~\cite{MY60}) showed how to construct a non deterministic finite automaton  with $n+1$ states where $n$ represents the number of  letters of a given regular expression.
The main idea of the construction  is to define  some particular sets named $\mathrm{First}$, $\mathrm{Follow}$ and $\mathrm{Last}$ that are computed with respect to the occurrences of the symbols that appear in the expression.

These so-called Glushkov automata (or Position automata) are finite state machines that have been deeply studied.
They have been structurally characterized by Caron and Ziadi~\cite{CZ00}, allowing us to invert the Glushkov computation by constructing an expression with $n$ symbols from a Glushkov automaton with $n+1$ states.
They have been considered too in the theoretical notion of one-unabiguity by Bruggemann-Klein and Wood~\cite{BKW98}, characterizing regular languages recognized by a deterministic Glushkov automaton, or with practical thoughts, like expression updating~\cite{BDHLM04}.
Finally, it is also related to combinatorial research topics.
As an example, Nicaud~\cite{Nic09} proved that the average number of transitions of Glushkov automata is linear.

Moreover, the Glushkov automata can be easily reduced into the Follow automata~\cite{IY03} in the case of word by applying an easy-to-compute congruence from the position functions.

The Glushkov construction was extended to tree automata~\cite{LOZ13,MOZ17}, using a Top-Down interpretation of tree expressions.
This interpretation can be problematic while considering determinism.
Indeed, it is a folklore that there exist regular tree languages that cannot be recognized by Top-Down deterministic tree automata.
Extensions of one-ambiguity are therefore incompatible with this approach.

In this paper, we propose a new approach based on the construction of
Glushkov and of Ilie and Yu in a Bottom-Up interpretation.
We also define a compressed version of tree automata in order to factorize the transitions, and we show how to apply it directly over these computations using natural factorizations due to the structure of the expressions.
The paper is structured as follows: in Section~\ref{sec:prelim}, we recall  some properties related to regular tree expressions; we also introduce  some basics definitions.
We define, in Section~\ref{sec:glushFun}, the position functions used  for the construction of the Bottom-Up position tree automaton.
Section~\ref{sec:BotUpAut} indicates the way  that we construct the  Bottom-Up position tree automaton with a linear number of states  using the functions shown in Section~\ref{sec:glushFun}.
In Section~\ref{sec:CompBotUpAut}, we propose the notion of compressed  automaton and show how to reduce the size of the Position automaton computed in the previous section.
Finally, we show in Section~\ref{sec:father} and in Section~\ref{sec:father-compressed} how to extend the notion of Follow automaton in a Bottom-Up interpretation.

This paper is an extended version of~\cite{AMZ18}.
It is a full-proof version that contains two new constructions: the father automaton and its compressed version.

\section{Premiminaries}\label{sec:prelim}
Let us first introduce some notations and preliminary definitions.
For a boolean condition $\psi$, we denote by $(E \mid \psi)$ $E$ if $\psi$ is satisfied, $\emptyset$ otherwise.
Let $\Sigma=(\Sigma_n)_{n\geq 0}$ be a finite ranked alphabet.
A \emph{tree} $t$ over $\Sigma$ is inductively defined by $t=f(t_1,\ldots,t_k)$ where $f\in\Sigma_k$ and $t_1,\ldots,t_k$  are $k$ trees over $\Sigma$.
The relation "$s$ is a subtree of $t$" is denoted by $s\prec t$ for any two trees $s$ and $t$.
We denote by $\mathrm{root}(t)$ the root symbol of the tree $t$, \emph{i.e.}
  \begin{equation}\label{eq def root tree}
    \mathrm{root}(f(t_1,\ldots,t_k))=f.
  \end{equation}
The \emph{predecessors} of a symbol $f$ in a tree $t$ are the symbols that appear directly above it.
We denote by $\mathrm{father}(t,f)$, for a tree $t$ and a symbol $f$ the pairs
    \begin{equation}\label{def fath}
      \mathrm{father}(t,f) = \{(g,i)\in\Sigma_l\times\mathbb{N} \mid \exists g(s_1,\ldots,s_l)\prec t, \mathrm{root}(s_i)=f\}.
    \end{equation}
These couples link the predecessors of $f$ and the indices of the subtrees in $t$ that $f$ is the root of.
Let us consider a tree $t=g(t_1,\ldots,t_k)$ and a symbol $f$.
By definition of the structure of a tree, a predecessor of $f$ in $t$ is a predecessor of $f$ in a subtree $t_i$ of $t$, or $g$ if $f$ is a root of a subtree $t_i$ of $t$.
Consequently:
  \begin{equation}\label{eq father tree}
    \mathrm{father}(t,f) = \bigcup_{i\leq n} \mathrm{father}(t_i,f) \cup\{(g,i)\mid f\in\mathrm{root}(t_i)\}.
  \end{equation}
We denote by $T_\Sigma$ the set of trees over $\Sigma$.
A tree language $L$ is a subset of $T_\Sigma$.

For any $0$-ary symbol $c$, let $t \cdot_c L$ denote the tree language constituted of the trees obtained by substitution of any symbol $c$ of $t$ by a tree of $L$.
By a linear extension, we denote by $L \cdot_c L' = \{t \cdot_c L' \mid t\in L \}$.
For an integer $n$, the $n$-th substitution $^{c,n}$ of a language $L$ is the language $L^{c,n}$ recursively defined by
\begin{align*}
  L^{c,n} &=
    \begin{cases}
      \{c\}, & \text{ if } n = 0,\\
      L \cdot_c L^{c, n - 1} & \text{ otherwise.}
    \end{cases}
\end{align*}
Finally, we denote by $L(E_1^{*_c})$ the language $ \bigcup_{k\geq 0} L(E_1)^{c,k}$.

An \emph{automaton} over $\Sigma$ is a $4$-tuple $\mathrm{A}=(Q,\Sigma,Q_F,\delta)$ where $Q$ is a set of states, $Q_F\subseteq Q$ is the set of final states, and $\delta\subset\bigcup_{k\geq 0} (Q^k\times \Sigma_k\times Q)$ is the set of transitions, which can be seen as the function from $Q^k \times \Sigma_k$ to $2^Q$ defined by
\begin{equation*}
  (q_1,\ldots,q_k,f,q) \in \delta \Leftrightarrow q \in \delta(q_1,\ldots,q_k,f).
\end{equation*}
It can be linearly extended as the function from $(2^{Q})^k \times \Sigma_k$ to $2^Q$ defined by
  \begin{equation}\label{eq:extDeltaEns}
    \delta(Q_1,\ldots,Q_n,f) = \displaystyle \bigcup_{(q_1,\ldots,q_n)\in Q_1\times\cdots Q_n} \delta(q_1,\ldots,q_n,f).
  \end{equation}
Finally, we also consider the function $\Delta$ from $T_{\Sigma}$ to $2^Q$ defined by
  \begin{equation*}
    \Delta(f(t_1,\ldots,t_n)) = \delta(\Delta(t_1),\ldots,\Delta(t_n),f).
  \end{equation*}
Using these definitions, the language $L(A)$ recognized by the automaton $A$ is the language
\begin{equation*}
  \{t\in T_{\Sigma} \mid \Delta(t)\cap Q_F \neq\emptyset\}.
\end{equation*}
A tree automaton $A=(\Sigma,Q,F,\delta)$ is \emph{deterministic} if for any symbol $f$ in $\Sigma_m$, for any $m$ states $q_1,\ldots,q_m$ in $Q$, $|\delta(q_1,\ldots,q_m,f)|\leq 1$.

\emph{A regular expression} $E$ over the alphabet $\Sigma$ is inductively defined by:
  \begin{align*}
    E&=f(E_1,\ldots,E_k),  & E &= E_1+E_2,\\
    E&=E_1\cdot_c E_2, &  E &=E_1^{*_c},
  \end{align*}
where $k\in\mathbb{N}$, $c\in\Sigma_0$, $f\in\Sigma_k$ and $E_1,\ldots,E_k$ are any $k$ regular expressions over $\Sigma$.
In what follows, we consider expressions where the subexpression $E_1\cdot_c E_2$ only appears when $c$ appears in the expression $E_1$.
The \emph{language denoted} by $E$ is the language $L(E)$ inductively defined by
  \begin{align*}
    L(f(E_1,\ldots,E_k)) &= \{f(t_1,\ldots,t_k)\mid t_j \in L(E_j), j\leq k\},\\
    L(E_1+E_2) &= L(E_1) \cup L(E_2),\\
    L(E_1\cdot_c E_2) &= L(E_1) \cdot_c L(E_2),\\
    L(E_1^{*_c}) &= L(E_1) ^{*_c},
  \end{align*}
with $k\in\mathrm{N}$, $c\in\Sigma_0$, $f\in\Sigma_k$ and $E_1,\ldots,E_k$ any $k$ regular expressions over $\Sigma$.

A regular expression $E$ is \emph{linear} if each symbol $\Sigma_{n}$ with $n\neq 0$ occurs at most once in $E$.
Note that the symbols of rank $0$ may appear more than once.
We denote by $\overline{E}$ the linearized form of $E$, which is the expression $E$ where any occurrence of a symbol is indexed by its position in the expression.
The set of indexed symbols, called \emph{positions}, is denoted by  $\mathrm{Pos}(\overline{E})$.
We also consider the \emph{delinearization} mapping $\mathrm{h}$ sending a linearized expression over its original unindexed version.

  Let $\phi$ be a function between two alphabets $\Sigma$ and $\Sigma'$ such that $\phi$ sends $\Sigma_n$ to $\Sigma'_n$ for any integer $n$.
  By a well-known adjunction, this function is extended to an \emph{alphabetical morphism} from $T(\Sigma)$ to $T(\Sigma')$ by setting $\phi(f(t_1,\ldots,t_n)) = \phi(f)(\phi(t_1),\ldots,\phi(t_n))$.
  As an example, one can consider the delinearization morphism $\mathrm{h}$ that sends an indexed alphabet to its unindexed version.
  Given a language $L$, we denote by $\phi(L)$ the set $\{\phi(t)\mid t\in L\}$.
  The \emph{image by} $\phi$ of an automaton $A=(\Sigma,Q,Q_F,\delta)$ is the automaton $\phi(A)=(\Sigma',Q,Q_F,\delta')$ where
  \begin{equation*}
    \delta' = \{(q_1,\ldots,q_n,\phi(f),q) \mid (q_1,\ldots,q_n,f,q) \in\delta\}.
  \end{equation*}
  By a trivial induction over the structure of the trees, it can be shown that
  \begin{equation}\label{eq:lienMorpAutLang}
    \phi(L(A)) = L(\phi(A)).
  \end{equation}
  An alphabetical morphism is a particular case of \emph{automaton morphism} between two automata $A=(\Sigma,Q,F,\delta)$ and $B=(\Sigma',Q',F',\delta')$, that is a function $\phi$ that sends $\Sigma_n$ to $\Sigma'_n$ for any integer $n$, $Q$ to $Q'$, $F$ to $F'$ and $\delta$ to $\delta'$ such that
  \begin{equation}\label{eq def morphism}
    \delta'((\phi(q_1),\ldots,\phi(q_n)),\phi(f)) = \{\phi(q) \mid q\in \delta((q_1,\ldots,q_n),f)\}.
  \end{equation}
  In this case, we set $\phi(A)=(\phi(\Sigma),\phi(Q),\phi(F),\phi(\delta))$.
  Let us first show that morphisms are stable w.r.t. transition composition.
  \begin{lemma}
    Let $A=(\Sigma,\_,\_,\delta)$ be an automaton and $\phi$ be an automaton morphism from $A$.
    Let $\phi(A)=(\Sigma',\_,\_,\delta')$.
    For any tree $t$ in $T(\Sigma)$,
    \begin{equation*}
      \Delta'(\phi(t)) = \{\phi(q) \mid q\in \Delta(t)\}.
    \end{equation*}
  \end{lemma}
  \begin{proof}
    Let us proceed by induction over the structure of the trees in $T(\Sigma)$.
    Let $t=f(t_1,\ldots,t_n)$ be a tree in $T(\Sigma)$.
    Let us set $S_i= \{\phi(q) \mid q\in \Delta(t_i)\}$ for any integer $i\leq n$.
    Notice that $S_i = (\Delta'(\phi(t_i))$ for any integer $i\leq n$ from the induction hypothesis.
    For a given state $q'$ of $\phi(A)$, we denote by $\phi^{-1}(q)$ the set of the states $q'$ in $A$ satisfying $\phi(q) = q'$.
    Then:
    from ,
    \begin{align*}
      \Delta'(\phi(f)(\phi(t_1),\ldots,\phi(t_n)))
        &= \delta'((\Delta'(\phi(t_1)),\ldots,\Delta'(\phi(t_n))),\phi(f)) \\
        &= \delta'((S_1,\ldots,S_n),\phi(f)) & \textbf{(Induction Hypothesis)}\\
        &= \bigcup_{(q_1,\ldots,q_n) \in S_1 \times \cdots \times S_n} \delta'((q_1,\ldots,q_n),\phi(f))\\
        &= \bigcup_{(p_1,\ldots,p_n) \in \Delta(t_1)\times\cdots\times\Delta(t_n)} \delta'((\phi(p_1),\ldots,\phi(p_n)),\phi(f))\\
        &= \bigcup_{(p_1,\ldots,p_n) \in \Delta(t_1)\times\cdots\times\Delta(t_n)} \{\phi(q) \mid q\in \delta((p_1,\ldots,p_n),f)\} & \textbf{(Equation~\eqref{eq def morphism})}\\
        &= \{\phi(q) \mid q\in \delta((\Delta(t_1),\ldots,\Delta(t_n)),f)\}\\
        &= \{\phi(q) \mid q\in \Delta(f(t_1,\ldots,t_n))\}.
    \end{align*}
    \qed
  \end{proof}
  As a direct consequence of the previous lemma,
  \begin{proposition}\label{prop morph pres lang}
    Let $A=(\Sigma,\_,\_,\_)$ be an automaton and $\phi$ be an automaton morphism from $A$.
    If for any symbol $f$ in $\Sigma$, $\phi(f)=f$, then
    \begin{equation*}
      L(A) = L(\phi(A)).
    \end{equation*}
  \end{proposition}
  Two automata $A$ and $B$ are \emph{isomorphic} if there exist two morphisms $\phi$ and $\phi'$ satisfying
  \begin{align*}
    A &= \phi'(\phi(A)), &
    B &= \phi(\phi'(B)).
  \end{align*}
\section{Position Functions}\label{sec:glushFun}
  In this section, we define the position functions that are considered in the construction of the Bottom-Up automaton in the next sections.
  We show how to compute them and how they characterize the trees in the language denoted by a given expression.

  Let $E$ be a linear expression over a ranked alphabet $\Sigma$ and $f$ be a symbol $\in \Sigma_k$.
  The set $\mathrm{Root}(E)$, subset of $\Sigma$, contains the roots of the trees in $L(E)$, \emph{i.e.}
  \begin{equation}\label{eq def root}
    \mathrm{Root}(E) = \{\mathrm{root}(t) \mid t\in L(E)\}.
  \end{equation}
  The set $\mathrm{Father}(E,f)$, subset of $\Sigma\times\mathbb{N}$, contains a couple $(g,i)$ if there exists a tree in $L(E)$ with a node labeled by $g$ the $i$-th child of is a node labeled by $f$:
  %
 \begin{equation}\label{eq def father}
    \mathrm{Father}(E,f)=\bigcup_{t\in L(E)} \mathrm{father}(t,f).
  \end{equation}
  \begin{example}\label{ex:calcul fonctions glushkov}
    Let us consider the ranked alphabet defined by $\Sigma_2=\{f\}$, $\Sigma_1=\{g\}$, and $\Sigma_0=\{a,b\}$.
    Let $E$ and $\overline{E}$ be the expressions defined by
    \begin{equation*}
      E = (f(a,a)+g(b))^{*_a}\cdot_b f(g(a),b), \quad
      \overline{E} = (f_1(a,a)+g_2(b))^{*_a}\cdot_b f_3(g_4(a),b).
    \end{equation*}
     Hence,
     \begin{gather*}
       \mathrm{Root}(\overline{E}) = \{a,f_1,g_2\},\\
       \begin{aligned}
         \mathrm{Father}(\overline{E},f_1) &= \{(f_1,1),(f_1,2)\}, &
         \mathrm{Father}(\overline{E},a) &= \{(f_1,1),(f_1,2),(g_4,1)\},\\
         \mathrm{Father}(\overline{E},g_2) &= \{(f_1,1),(f_1,2)\}, &
         \mathrm{Father}(\overline{E},b) &= \{(f_3,2)\},\\
         \mathrm{Father}(\overline{E},f_3) &= \{(g_2,1)\}, &
         \mathrm{Father}(\overline{E},g_4) &= \{(f_3,1)\}.
       \end{aligned}
     \end{gather*}
 \end{example}
  Let us show how to inductively compute these functions.
  \begin{lemma}\label{lem root ind}
    Let $E$ be a linear expression over a ranked alphabet $\Sigma$.
    The set $\mathrm{Root}(E)$ is inductively computed as follows:
    \begin{align*}
      \mathrm{Root}(f(E_1,...,E_n)) &= \{f\},\\
      \mathrm{Root}(E_1+E_2) &= \mathrm{Root}(E_1)\cup \mathrm{Root}(E_2),\\
      \mathrm{Root}(E_1\cdot_c E_2) &=
        \begin{cases}
          \mathrm{Root}(E_1)\setminus\{c\}\cup \mathrm{Root}(E_2)& \text{ if } c\in L(E_1),\\
          \mathrm{Root}(E_1) & \text{otherwise,}
        \end{cases}\\
      \mathrm{Root}(E_1^{*_c}) &= \mathrm{Root}(E_1)\cup\{c\},
    \end{align*}
    where $E_1,\ldots,E_n$ are $n$ regular expressions over $\Sigma$, $f$ is a symbol in $\Sigma_n$ and $c$ is a symbol in $\Sigma_0$.
  \end{lemma}
  \begin{proof}
    Let us consider the following cases.

    \begin{enumerate}
      \item The case when $E=f(E_1,\ldots,E_n)$ is a direct consequence of Equation~\eqref{eq def root tree} and Equation~\eqref{eq def root}.
      \item Let us consider a tree $t$ in $L(E_1+E_2)$.
        Consequently, $t$ is in $L(E_1)$ or in $L(E_2)$.
        And we can conclude using Equation~\eqref{eq def root}.
      \item Let us consider a tree $t$ in $t_1 \cdot_c L(E_2)$ with $t_1\in L(E_1)$.
        If $t_1 = c$ (resp. $t_1 \neq c$), then it holds by definition of $\cdot_c$ that $\mathrm{root}(t)\in\mathrm{Root}(L(E_2))$ (resp. $\mathrm{root}(t)=\mathrm{root}(t_1)\setminus \{c\}$).
        Once again, we can conclude using Equation~\eqref{eq def root}.
      \item Let us consider a tree $t$ in $L(E_1^{*_c})$.
        By definition, either $t=c$ or $t$ is in $t_1\cdot_c  L(E_1^{*_c})$ with $t_1\in L(E_1)\setminus\{c\}$.
        Therefore either $\mathrm{root}(t)= c$ or $\mathrm{root}(t)=\mathrm{root}(t_1)$.
        Once again, we can conclude using Equation~\eqref{eq def root}.
    \end{enumerate}
    \qed
  \end{proof}

  \begin{lemma}\label{lem father ind}
    Let $E$ be a linear expression and $f$ be a symbol in $\Sigma_k$.
    The set $\mathrm{Father}(E,f)$ is inductively computed as follows:
    \begin{align*}
      \mathrm{Father}(g(E_1,...,E_n),f) &= \bigcup_{i\leq n} \mathrm{Father}(E_i,f) \cup\{(g,i)\mid f\in\mathrm{Root}(E_i)\},\\
      \mathrm{Father}(E_1+E_2,f) &= \mathrm{Father}(E_1,f)\cup \mathrm{Father}(E_2,f),\\
      \mathrm{Father}(E_1\cdot_c E_2,f) &=
          (\mathrm{Father}(E_1,f) \mid f \neq c) \cup \mathrm{Father}(E_2,f)\\
        & \qquad \cup (\mathrm{Father}(E_1,c) \mid f\in\mathrm{Root}(E_2))\\
      \mathrm{Father}(E_1^{*_c},f) &= \mathrm{Father}(E_1,f) \cup (\mathrm{Father}(E_1,c) \mid f\in\mathrm{Root}(E_1)),\\
    \end{align*}
    where $E_1,\ldots,E_n$ are $n$ regular expressions over $\Sigma$, $g$ is a symbol in $\Sigma_n$ and $c$ is a symbol in $\Sigma_0$.
  \end{lemma}
  \begin{proof}
    Let us consider the following cases.
    \begin{enumerate}
      \item The case when $E=g(E_1,\ldots,E_n)$ is a direct consequence of Equation~\eqref{eq father tree} and Equation~\eqref{eq def root}.
      \item Let us consider a tree $t$ in $L(E_1+E_2)$.
        Consequently, $t$ is in $L(E_1)$ or in $L(E_2)$.
        And we can conclude using Equation~\eqref{eq father tree}.
      \item Let us consider a tree $t=t_1 \cdot_c L(E_2)$ with $t_1\in L(E_1)$.
        By definition, $t$ equals $t_1$ where the occurrences of $c$ have been replaced by some trees $t_2$ in $L(E_2)$.
        Two cases may occur.
        \begin{enumerate}
          \item If $c \neq f$, then a predecessor of the symbol $f$ in $t$ can be a predecessor of the symbol $f$ in a tree $t_2$ in $L(E_2)$, a predecessor of the symbol $f$ in $t_1$, or a predecessor of $c$ in $t_1$ if an occurrence of $c$ in $t_1$ has been replaced by a tree $t_2$ in $L(E_2)$ the root of which is $f$.
          \item If $c = f$, since the occurrences of $c$ have been replaced by some trees $t_2$ of $L(E_2)$, a predecessor of the symbol $c$ in $t$ can be a predecessor of the symbol $c$ in a tree $t_2$ in $L(E_2)$, or a predecessor of $c$ in $t_1$ if an occurrence of $c$ has been replaced by itself (and therefore if it appears in $L(E_2)$).
        \end{enumerate}
        And we can conclude using Equation~\eqref{eq father tree}  and Equation~\eqref{eq def root}.
      \item By definition, $L(E_1^{*_c}) = \bigcup_{k\geq 0} L(E_1)^{c,k}$.
        Therefore, a tree $t$ in $L(E_1^{*_c})$ is either $c$ or a tree $t_1$ in $L(E_1)$ where the occurrences of $c$ have been replaced by some trees $t_2$ in $L(E_1)^{c,k}$ for some integer $k$.
        Let us then proceed by recursion over this integer $k$.
        If $k = 1$, a predecessor of $f$ in $t$ is a predecessor of $f$ in $t_1$, a predecessor of $f$ in a tree $t_2$ in $L(E_1)^{c,1}$ or a predecessor of $c$ in $t_1$ if an occurrence of $c$ in $t_1$ was substituted by a tree $t_2$ in $L(E_1)^{c,1}$ the root of which is $f$, \emph{i.e.}
        \begin{equation*}
            \mathrm{Father}(E_1^{c,2},f) = \mathrm{Father}(E_1,f) \cup (\mathrm{Father}(E_1,c) \mid f\in\mathrm{Root}(E_1)).
        \end{equation*}
        By recursion over $k$ and by applying the same reasoning, it can be shown that each recursion step adds $\mathrm{Father}(E_1,f)$ to the result of the previous step, and therefore
        \begin{equation*}
            \mathrm{Father}(E_1^{c,k},f) = \mathrm{Father}(E_1,f) \cup (\mathrm{Father}(E_1,c) \mid f\in\mathrm{Root}(E_1)).
        \end{equation*}
    \end{enumerate}
    \qed
  \end{proof}
  Let us now show how these functions characterize, for a tree $t$, the membership of $t$ in the language denoted by an expression.
  \begin{definition}
    Let $E$ be a linear expression over a ranked alphabet $\Sigma$ and $t$ be a tree in $T(\Sigma)$.
    The property $P(t)$ is the property defined by
    \begin{equation*}
      \forall s=f(t_1,\ldots,t_n) \prec t, \forall i\leq n, (f,i)\in \mathrm{Father}(E,\mathrm{root}(t_i)).
    \end{equation*}
  \end{definition}
  \begin{proposition}\label{prop caract tree in lang}
    Let $E$ be a linear expression over a ranked alphabet $\Sigma$ and $t$ be a tree in $T(\Sigma)$.
    Then \textbf{(1)} $t$ is in $L(E)$ if and only if \textbf{(2)} $\mathrm{root}(t)$ is in $\mathrm{Root}(E)$ and $P(t)$ is satisfied.
  \end{proposition}
  \begin{proof}
    Let us first notice that the proposition $1 \Rightarrow 2$ is direct by definition of $\mathrm{Root}$ and $\mathrm{Father}$.
    Let us show the second implication by induction over the structure of $E$.
    Hence, let us suppose that $\mathrm{root}(t)$ is in $\mathrm{Root}(E)$ and $P(t)$ is satisfied.
    \begin{itemize}
      \item Let us consider the case when $E=g(E_1,\ldots,E_n)$ and let us set $t=$ $f(t_1,\ldots,$ $t_n)$.
        Since $\mathrm{root}(t)$ is in $\mathrm{Root}(E)$, $f = g$ from Lemma~\ref{lem root ind}.
        From $P(t)$, it holds that for any $i\leq n$, $(f,i)\in \mathrm{Father}(E,\mathrm{root}(t_i))$.
        Since $E$ is linear, and following Lemma~\ref{lem father ind}, $\mathrm{root}(t_i)\in\mathrm{Root}(E_i)$.
        Consequently, from the induction hypothesis, $t_i$ is in $L(E_i)$ for any integer $i\leq n$ and $t$ belongs to $L(E)$.
      \item The case of the sum is a direct application of the induction hypothesis.
      \item Let us consider the case when $E=E_1\cdot_c E_2$.
        Let us first suppose that $\mathrm{root}(t)$ is in $\mathrm{Root}(E_2)$.
        Then $c$ is in $L(E_1)$ and $P(t)$ is equivalent to
        \begin{equation*}
          \forall s=f(t_1,\ldots,t_n) \prec t, \forall i\leq n, (f,i)\in \mathrm{Father}(E_2,\mathrm{root}(t_i)).
        \end{equation*}
        By induction hypothesis $t$ is in $L(E_2)$ and therefore in $L(E)$.

        Let us suppose now that $\mathrm{root}(t)$ is in $\mathrm{Root}(E_1)$.
        Since $E$ is linear, let us consider the subtrees $t_2$ of $t$ with only symbols of $E_2$ and a symbol of $E_1$ as a predecessor in $t$.
        Since $P(t)$ holds, according to induction hypothesis and Lemma~\ref{lem father ind}, each of these trees belongs to $L(E_2)$.
        Hence $t$ belongs to $t_1 \cdot_c L(E_2)$ where $t_1$ is equal to $t$ where the previously defined $t_2$ trees are replaced by $c$.
        Once again, since $P(t)$ holds and since $\mathrm{root}(t)$ is in $\mathrm{Root}(E_1)$, $t_1$ belongs to $L(E_1)$.

        In these two cases, $t$ belongs to $L(E)$.
      \item Let us consider the case when $E=E_1^{*_c}$.
        Let us proceed by induction over the structure of $t$.
        If $t=c$, the proposition holds from Lemma~\ref{lem root ind} and Lemma~\ref{lem father ind}.
        Following Lemma~\ref{lem father ind}, each predecessor of a symbol $f$ in $t$ is a predecessor of $f$ in $E_1$ (case \textbf{1}) or a predecessor of $c$ in $E_1$ (case \textbf{2}).
        If all the predecessors of the symbols satisfy the case \textbf{1}, then by induction hypothesis $t$ belongs to $L(E_1)$ and therefore to $L(E)$.
        Otherwise, we can consider (similarly to the catenation product case) the smallest subtrees $t_2$ of $t$ the root of which admits a predecessor in $t$ which is a predecessor of $c$ in $E_1$.
        By induction hypothesis, these trees belong to $L(E_1)$.
        And consequently $t$ belongs to $t' \cdots L(E_1)$ where $t'$ is equal to $t$ where the subtrees $t_2$ have been substituted by $c$.
        Once again, by induction hypothesis, $t'$ belongs to $L(E_1^{*_c})$.
        As a direct consequence, $t$ belongs to $L(E)$.
    \end{itemize}
    \qed
  \end{proof}
\section{Bottom-Up Position Automaton}\label{sec:BotUpAut}

  In this section, we show how to compute a Bottom-Up automaton with a linear number of states from the position functions previously defined.
  \begin{definition}\label{def glush bu}
    The \emph{Bottom-Up Position automaton} $\mathcal{P}_{E}$ of a linear expression $E$ over a ranked alphabet $\Sigma$ is the automaton $(\Sigma,\mathrm{Pos}(E),\mathrm{Root}(E),\delta)$ defined by:
    \begin{equation*}
      ((f_1,\ldots,f_n),g,g) \in \delta  \Leftrightarrow \forall i \leq n,      (g,i)\in\mathrm{Father}(E,f_i).
    \end{equation*}
  \end{definition}
  Notice that due to the linearity of $E$, $\mathcal{P}_{E}$ is deterministic.
  \begin{example}
    The Bottom-Up Position automaton $(\mathrm{Pos}(\overline{E}),\mathrm{Pos}(\overline{E}),\mathrm{Root}(\overline{E}),\delta)$ of the expression $\overline{E}$ defined in Example~\ref{ex:calcul fonctions glushkov} is defined as follows:
    \begin{gather*}
      \begin{aligned}
        \mathrm{Pos}(E) & = \{a,b,f_1,g_2,f_3,g_4\}, &
      \mathrm{Root}(\overline{E}) &= \{a,f_1,g_2\},
      \end{aligned}\\
      \begin{aligned}
        \delta &= \{(a,a), (b,b),  ((a,a),f_1,f_1), ((a,f_1),f_1,f_1), ((a,g_2),f_1,f_1), ((f_1,a),f_1,f_1),\\
        & \qquad  ((f_1,f_1),f_1,f_1), ((f_1,g_2),f_1,f_1), ((g_2,a),f_1,f_1), ((g_2,f_1),f_1,f_1),\\
        & \qquad  ((g_2,g_2),f_1,f_1), (f_3, g_2,g_2),((b,g_4),f_3,f_3), (a,g_4,g_4)\}.
      \end{aligned}
    \end{gather*}
  \end{example}
  Let us now show that the Position automaton of $E$ recognizes $L(E)$.
  \begin{lemma}\label{lem lien Delta P}
    Let $\mathcal{P}_E=(\Sigma,Q,Q_F,\delta)$ be the Bottom-Up Position automaton of a linear expression $E$ over a ranked alphabet $\Sigma$, $t$ be a tree in $T_\Sigma$ and $f$ be a symbol in $\mathrm{Pos}(E)$.
    Then
    \textbf{(1)} $f\in\Delta(t)$ if and only if \textbf{(2)} $\mathrm{root}(t) = f \wedge P(t)$.
  \end{lemma}
  \begin{proof}
    Let us proceed by induction over the structure of $t=f(t_1,\ldots,t_n)$.
    By definition, $\Delta(t) = \delta(\Delta(t_1),\ldots,\Delta(t_n),f)$.
    For any state $f_i$ in $\Delta_i$, it holds from the induction hypothesis that
    \begin{equation}\label{eq:hypIndLemAutPosLienP}
      f_i\in\Delta(t_i) \Leftrightarrow \mathrm{root}(t_i) = f_i \wedge P(t_i).\tag{*}
    \end{equation}
    Then, suppose that \textbf{(1)} holds (\emph{i.e.} $f\in\Delta(t)$).
    Equivalently, there exists by definition of $\mathcal{P}_E$ a transition $((f_1,\ldots,f_n),f,f)$ in $\delta$ such that $f_i$ is in $\Delta(t_i)$ for any integer $i\leq n$.
    Consequently, $f$ is the root of $t$.
    Moreover, from the equivalence stated in Equation~\eqref{eq:hypIndLemAutPosLienP}, $\mathrm{root}(t_i) = f_i $ and $P(t_i)$ holds for any integer $i\leq n$.
    Finally and equivalently, $P(t)$ holds as a consequence of Equation~\eqref{eq father tree}.
    The reciprocal condition can be proved similarly since only equivalences are considered.
    \qed
  \end{proof}
  As a direct consequence of Lemma~\ref{lem lien Delta P} and Proposition~\ref{prop caract tree in lang},
  \begin{proposition}
    The Bottom-Up Position automaton of a linear expression $E$ recognizes $L(E)$.
  \end{proposition}
  The Bottom-Up Position automaton of a (not necessarily linear) expression $E$ can be obtained by first computing the Bottom-Up Position automaton of its linearized expression $\overline{E}$ and then by applying the alphabetical morphism $\mathrm{h}$.
  As a direct consequence of Equation~\eqref{eq:lienMorpAutLang},
  \begin{proposition}
    The Bottom-Up Position automaton of an expression $E$ recognizes $L(E)$.
  \end{proposition}

\section{Compressed Bottom-Up Position Automaton}\label{sec:CompBotUpAut}
  In this section, we show that the structure of an expression allows us to factorize the transitions of a tree automaton by only considering the values of the $\mathrm{Father}$ function.
  The basic idea of the factorizations is to consider the cartesian product of sets.
  Imagine that a tree automaton contains four binary transitions $(q_1,q_1,f,q_3)$, $(q_1,q_2,f,q_3)$, $(q_2,q_1,f,q_3)$ and $(q_2,q_2,f,q_3)$.
  These four transitions can be factorized as a \emph{compressed transition} $(\{q_1,q_2\},\{q_1,q_2\},f,q_3)$ using set of states instead of sets.
  The behavior of the original automaton can be simulated by considering the cartesian product of the origin states of the transition.

  We first show how to encode such a notion of compressed automaton and how it can be used in order to solve the membership test.

  \begin{definition}
    A \emph{compressed tree automaton} over a ranked alphabet $\Sigma$ is a $4$-tuple $(\Sigma,Q,Q_F,\delta)$ where $Q$ is a set \emph{of states}, $Q_F\subset Q$ is the set of \emph{final states}, $\delta \subset (2^Q)^n\times \Sigma_n\times 2^Q$ is the set of \emph{compressed transitions} that can be seen as a function from $(2^{Q})^k\times\Sigma_k$ to $2^Q$ defined by
    \begin{equation*}
      (Q_1,\ldots,Q_k,f,q)\in\delta \Leftrightarrow q\in \delta(Q_1,\ldots,Q_k,f).
    \end{equation*}
  \end{definition}
  \begin{example}\label{ex:aut compr}
  Let us consider the compressed automaton $A=(\Sigma,Q,Q_F,\delta)$ shown in Figure~\ref{exp comp aut}.
  Its transitions are
  \begin{equation*}
    \delta = \{ (\{1,2,5\},\{3,4\},f,1), (\{2,3,5\},\{4,6\},f,2),
      (\{1,2\},\{3\},f,5),(\{6\},g,4), (\{6\},g,5), (a,6), (a,4), (b,3)\}.
  \end{equation*}
\end{example}

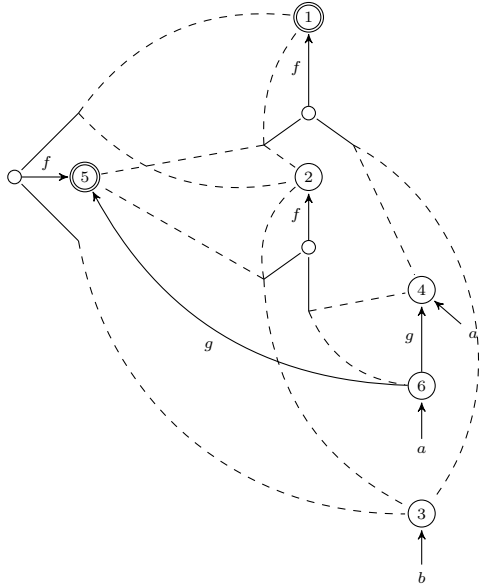
\begin{figure}[H]
    \centerline{
    \begin{tikzpicture}[node distance=2.5cm,bend angle=30,transform shape,scale=0.85]
      \node[state] (2)  {$2$} ;
     \node[state, double,above of=2, node distance = 2.5cm] (1) {$1$};
     \node[state,double,  left  of=2,node distance=3.5cm] (5) {$5$};
      \node[state, below right of=2,node distance = 2.5cm] (4) {$4$};
      \node[state, below of=4, node distance = 3.5cm] (3) {$3$};
      \node[state,below  of=4,node distance=1.5cm] (6) {$6$};
       \node[state, below of=1,node distance=1.5cm] (cerc) {};
       \node[state,below of=2,node distance=1.1cm] (cerc1) {};
       \node[,state,left of=5,node distance=1.1cm] (cerc2) {};
        \draw (3) ++(0cm,-1cm) node {$b$}  edge[->] (3);
        \draw (6) ++(0cm,-1cm) node {$a$}  edge[->] (6);
        \draw[left] (4) ++(1cm,-0.7cm) node {$a$}  edge[above,->] (4);
         \path[->]
       (cerc) edge[->, left] node {$f$} (1)
       (cerc1) edge[->, left] node {$f$} (2)
       (6) edge[->, left] node {$g$} (4)
       (6) edge[->, below left,bend left] node {$g$} (5)
       (cerc2) edge[->, above] node {$f$} (5);
       \draw (cerc) ++(-0.7 cm,-0.5cm)edge node [above,pos=0.5] {} (cerc)  edge[dashed]node[left,above,pos=0.6]{} (5)  edge[dashed] node[left,above,pos=0.5]{}(2) edge[dashed,bend left=20] node[bend left=20]{}(1);
       \draw (cerc) ++(0.7 cm,-0.5cm)edge node [above,pos=0.5] {} (cerc)  edge[dashed,bend left=50 ]node[left,above,pos=0.6]{} (3)  edge[dashed] node[left,above,pos=0.5]{}(4);
      \draw (cerc1) ++(-0.7 cm,-0.5cm)edge node [above,pos=0.5] {} (cerc1)  edge[dashed,bend left=30 ]node[left,above,pos=0.6]{} (2)  edge[dashed] node[left,above,pos=0.5]{}(5) edge[dashed,bend right=30 ]node[left,above,pos=0.6]{} (3)  ;
      \draw (cerc1) ++(0cm,-1cm)edge node [right,pos=0.5] {} (cerc1)  edge[dashed,bend right=30 ]node[left,above,pos=0.6]{} (6)
       edge[dashed] node[left,above,pos=0.5]{}(4);
       \draw (cerc2) ++(1cm,1cm)edge node [left,pos=0.5] {} (cerc2)  edge[dashed,bend left=30 ]node[left,above,pos=0.6]{} (1)  edge[dashed,bend right=30 ]node[left,above,pos=0.6]{} (2);
       \draw (cerc2) ++(1cm,-1cm)edge node [left,pos=0.5] {} (cerc2)  edge[dashed,bend right=40 ]node[left,above,pos=0.6]{} (3);
    \end{tikzpicture}}
    \caption{ The compressed automaton $A$.}
    \label{exp comp aut}
\end{figure}
  The transition function $\delta$ can be restricted to a function from $Q^n\times \Sigma_n$ to $2^Q$ (\emph{e.g.} in order to simulate the behavior of an uncompressed automaton) by considering for a tuple $(q_1,\ldots,q_k)$ of states and a symbol $f$ in $\Sigma_k$ all the "active" transitions $(Q_1,\ldots,Q_k,f,q)$, that are the transitions where $q_i$ is in $Q_i$ for $i\leq k$.
  More formally, for any states $(q_1,\ldots,q_k)$ in $Q^k$, for any symbol $f$ in $\Sigma_k$,
    \begin{equation}\label{eq:extdeltaEnsComp}
      \delta(q_1,\ldots,q_k,f) = \bigcup_{\substack{(Q_1,\ldots, Q_k, f,q)\in\delta,\\ \forall i\leq k, q_i\in Q_i}} \{q\}.
    \end{equation}
  The transition set $\delta$ can be extended to a function $\Delta$ from $T(\Sigma)$ to $2^Q$ by inductively considering, for a tree $f(t_1,\ldots,t_k)$ the "active" transitions $(Q_1,\ldots,Q_k,$ $f,q)$ once a subtree is read, that is when $\Delta(q_i)$ and $Q_i$ admits a common state for $i\leq k$.
  More formally, for any tree $t=f(t_1,\ldots,t_k)$ in $T(\Sigma)$,
  \begin{equation*}
    \Delta(t) = \bigcup_{\substack{(Q_1,\ldots,Q_k,f,q)\in\delta,\\ \forall i\leq k, \Delta(t_i)\cap Q_i\neq\emptyset}} \{q\}.
  \end{equation*}
  As a direct consequence of the two previous equations,
  \begin{equation}\label{eq:lienExtsTransComp}
    \Delta(f(t_1,\ldots,t_n)) = \bigcup_{(q_1,\ldots,q_n)\in\Delta(t_1)\times\cdots\times\Delta(t_n)} \delta(q_1,\ldots,q_n,f).
  \end{equation}
  The \emph{language recognized by} a compressed automaton $A=(\Sigma,Q,Q_F,\delta)$ is the subset $L(A)$ of $T(\Sigma)$ defined by
  \begin{equation*}
    L(A) = \{t\in T(\Sigma) \mid \Delta(t)\cap Q_F\neq\emptyset\}.
  \end{equation*}
  \begin{example}
    Let us consider the automaton of Figure~\ref{exp comp aut} and let us show that the tree $t=f(f(b,a),g(a))$ belongs to $L(A)$.
    In order to do so, let us compute $\Delta(t')$ for each subtree $t'$ of $t$.
    First, by definition,
    \begin{align*}
      \Delta(a)&=\{4,6\}, & \Delta(b)&=\{3\}.
    \end{align*}
    Since the only transition in $\delta$ labeled by $f$ containing $3$ in its first origin set and $4$ or $6$ in its second is the transition $(\{2,3,5\},\{4,6\},f,2)$,
    \begin{equation*}
      \Delta(f(b,a)) = \{2\}.
    \end{equation*}
    Since the two transitions labeled by $g$ are $(\{6\},g,4)$ and $(\{6\},g,5)$,
    \begin{equation*}
      \Delta(g(a)) = \{4,5\}.
    \end{equation*}
    Finally, there are two transitions labeled by $f$ containing $2$ in their first origin and $4$ or $5$ in its second: $(\{2,3,5\},\{4,6\},f,2)$ and $(\{1,2,5\},\{3,4\},f,1)$.
    Therefore
    \begin{equation*}
      \Delta(f(f(b,a),g(a)) = \{1,2\}.
    \end{equation*}
    Finally, since $1$ is a final state, $t\in L(A)$.
  \end{example}
  Let $\phi$ be an alphabetical morphism between two alphabets $\Sigma$ and $\Sigma'$.
  The \emph{image by} $\phi$ of a compressed automaton $A=(\Sigma,Q,Q_F,\delta)$ is the compressed automaton $\phi(A)=(\Sigma',Q,Q_F,\delta')$ where
  \begin{equation*}
    \delta'=\{(Q_1,\ldots,Q_n,\phi(f),q)\mid (Q_1,\ldots,Q_n,f,q) \in \delta\}.
  \end{equation*}
  By a trivial induction over the structure of the trees, it can be shown that
  \begin{equation}\label{eq:lien morph lang compress}
    L(\phi(A)) = \phi(L(A)).
  \end{equation}
  Due to their inductive structure, regular expressions are naturally factorizing the structure of transitions of a Glushkov automaton.
  Let us now define the compressed Position automaton of an expression.
  \begin{definition}\label{def glu comp}
    The \emph{compressed Bottom-Up Position automaton} $\mathcal{C}(E)$ of a linear expression $E$ is the automaton $(\Sigma,\mathrm{Pos}(E),\mathrm{Root}(E),\delta)$ defined by
    \begin{equation*}
      \delta = \{(Q_1,\ldots,Q_k,f,\{f\}) \mid Q_i = \{g \mid (f,i) \in \mathrm{Father}(E,g)\} \}.
    \end{equation*}
  \end{definition}
  \begin{example}
    Let us consider the expression $\overline{E}$ defined in Example~\ref{ex:calcul fonctions glushkov}.
    The compressed automaton of $\overline{E}$ is represented at Figure~\ref{linear_compressed_automata}.
  \end{example}
\begin{figure}[H]
  \centerline{
  \begin{tikzpicture}[node distance=2.5cm,bend angle=30,transform shape,scale=1]
    \node[accepting ,state] (f1)  {$f_1$} ;
    \node[state,double,  above right of=f1, node distance =2cm] (g_2) {$g_2$};
     \node[state, below right of=g_2, node distance = 3cm] (f_3) {$f_3$};
    \node[state, below of= f_3,node distance=1.7cm] (b) {$b$};
    \node[state,  below left of=f_3] (g_4) {$g_4$};
    \node[state,  below of=f1,node distance = 1cm] (cerc) {};
    \node[state,  left of=f_3,node distance = 1.3cm] (cerc1) {};
    \node[state, double,  below left of=f1,node distance = 2.5cm] (a) {$a$};
    \draw (a) ++(-1cm,0cm) node {$a$}  edge[->] (a);
    \draw (b) ++(0cm,-1cm) node {$b$}  edge[->] (b);
    \path[->]
     (f_3) edge[->,  bend right,right] node {$g_2$} (g_2)
     (a) edge[->,below left,bend right] node {$g_4$} (g_4)
     (cerc) edge[->, left] node {$f_1$} (f1)
     (cerc1) edge[->, above] node {$f_3$} (f_3);
    \draw (cerc) ++(0.5 cm,-0.5cm)edge node [above,pos=0.5] {} (cerc)  edge[dashed]node[left,above,pos=0.6]{} (g_2)  edge[dashed,bend left=10] node[left,above,pos=0.5]{}(a)edge[dashed,bend right=30] node[left,above,pos=0.5]{}(f1);
    \draw (cerc) ++(-0.5 cm,-0.5cm)edge node [above,pos=0.5] {} (cerc)  edge[dashed,bend left=70]node[left,above,pos=0.6]{} (g_2)  edge[dashed] node[left,above,pos=0.5]{}(a)
    edge[dashed,bend left=30] node[left,above,pos=0.5]{}(f1);
    \draw (cerc1) ++(-0.5 cm,-0.5cm)edge node [above,pos=0.5] {} (cerc1)  edge[dashed,bend right]node[left,above,pos=0.6]{} (g_4);
    \draw (cerc1) ++(0.5 cm,-0.5cm)edge node [above,pos=0.5] {} (cerc1)  edge[dashed,bend left=15]node[left,above,pos=0.6]{} (b);
    \end{tikzpicture}}
    \caption{The compressed automata of the expression $(f_1(a,a)+g_2(b))^{*_a}\cdot_b f_3(g_4(a),b)$.}
    \label{linear_compressed_automata}
\end{figure}
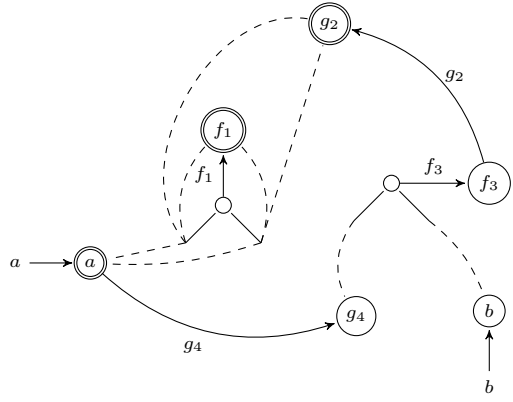
  As a direct consequence of Definition~\ref{def glu comp} and of Equation~\eqref{eq:extdeltaEnsComp},
  \begin{lemma}\label{lem comp rest trans state}
    Let $E$ be a linear expression over a ranked alphabet $\Sigma$.
    Let $\mathcal{C}(E) =(\Sigma,Q,Q_F,\delta)$.
    Then, for any states $(q_1,\ldots,q_n)$ in $Q^n$, for any symbol $f$ in $\Sigma_k$,
    \begin{equation*}
      \delta(q_1,\ldots,q_n,f) = \{f\} \Leftrightarrow \forall i \leq n, (f,i) \in \mathrm{Father}(E,q_i).
    \end{equation*}
  \end{lemma}
  Consequently, considering Definition~\ref{def glush bu}, Lemma~\ref{lem comp rest trans state} and Equation~\eqref{eq:lienExtsTransComp},
  \begin{proposition}\label{prop mem chemin compress glu}
    Let $E$ be a linear expression over a ranked alphabet $\Sigma$.
    Let $\mathcal{P}_E=(\_,\_,\_,\delta)$ and $\mathcal{C}(E) =(\_,\_,\_,\delta')$.
    For any tree $t$ in $T(\Sigma)$,
    \begin{equation*}
      \Delta(t) = \Delta'(t).
    \end{equation*}
  \end{proposition}
  Since the Bottom-Up Position automaton of a linear expression $E$ and its compressed version have the same states and the same final states,
  \begin{corollary}
    The Glushkov automaton of an expression and its compact version recognize the same language.
  \end{corollary}
  The compressed Bottom-Up Position automaton of a (not necessarily linear) expression $E$ can be obtained by first computing the compressed Bottom-Up Position automaton of its linearized expression $\overline{E}$ and then by applying the alphabetical morphism $\mathrm{h}$.
  Therefore, considering Equation~\eqref{eq:lien morph lang compress},
  \begin{proposition}
    The compressed Bottom-Up Position automaton of a regular expression $E$ recognizes $L(E)$.
  \end{proposition}

\section{The Father Automaton}\label{sec:father}

  In this section, we define the Father automaton associated with an expression that is an extension of the classical follow (word) automaton~\cite{IY03} that have been already extended in the case of (Top-Down) tree automata~\cite{MOZ17}.

  We embed the computation of the function $\mathrm{Root}$ in the function $\mathrm{Father}$ by adding a unary symbol $\$$ that is not in $\Sigma$ at the top of the syntactic tree of an expression.
  Indeed,
  \begin{equation}
    f \in \mathrm{Root}(E) \Leftrightarrow (\$, 1) \in \mathrm{Father}(\$(E),f).
  \end{equation}
  Equivalenty,
  \begin{lemma}\label{lem eq fin state dollar father}
    Let $f$ be a state of the Position automaton of a linear expression $E$.
    The two following conditions hold:
    \begin{enumerate}
      \item $f$ is a final state,
      \item $(\$, 1)$ is in $\mathrm{Father}(\$(E),f)$.
    \end{enumerate}
  \end{lemma}
  With this notation, the Father automaton is defined as follows.
  \begin{definition}\label{def father aut}
    The \emph{Father Automaton} of a linear expression $E$ over an alphabet $\Sigma$ is the automaton $\mathcal{F}_E=(\Sigma,Q,F,\delta)$ defined by
    \begin{gather*}
      \begin{aligned}
        Q &= \{\mathrm{Father}(\$(E),f) \mid f\in\Sigma\}, &
        F &= \{q\in Q \mid \$ \in q\},
      \end{aligned}\\
      ((\mathrm{Father}(\$(E),f_1),\ldots,\mathrm{Father}(\$(E),f_n)),g,\mathrm{Father}(\$(E),g)) \in \delta  \Leftrightarrow \forall i \leq n,      (g,i)\in\mathrm{Father}(E,f_i).
    \end{gather*}
  \end{definition}
  Notice that due to the linearity of $E$, $\mathcal{P}_{E}$ is deterministic.

  In the word case, it has been shown that the Follow automaton is a quotient of the Position automaton.
  Let us proceed in the same way in order to show, using Proposition~\ref{prop morph pres lang}, that the Father automaton of an expression $E$ recognizes $L(E)$.
  Consequently, let us first extend the notion of congruence for tree automata.

  Let $A=(\Sigma, Q, F, \delta)$ be a deterministic tree automaton and $\sim$ be an equivalence relation over $Q$ such that for any two equivalent states $p$ and $p'$,
  \begin{equation*}
    p\in F \Leftrightarrow p'\in F.
  \end{equation*}
  Given an equivalence relation $\sim$ over a set $S$, we denote by $S_\sim$ the set of equivalence classes of $S$ and by $[p]_\sim$ (or $[p]$ when there is no ambiguity) the equivalence classe of an element $p$ in $P$.
  By a notation abuse, we extend $\sim$ to the subsets of size $0$ or $1$ as follows:
  \begin{gather*}
    \begin{aligned}
      \emptyset &\sim \emptyset, & \forall s\in S, \emptyset \not\sim \{s\},\\
    \end{aligned}\\
    \forall s,s'\in S, \{s\} \sim \{s'\} \Leftrightarrow s \sim s'.
  \end{gather*}
  The relation $\sim$ is a \emph{Bottom-Up congruence} for $\delta$ if and only if for any two states $p$ and $p'$ in $Q$, the two following conditions are equivalent:
    \begin{enumerate}
      \item $p\sim p'$,
      \item for any symbol $f$ in $\Sigma_m$, for any integer $n\leq m$, for any $m-1$ states $q_1$, $\ldots$, $q_{n-1}$, $q_{n+1}$, $\ldots$, $q_m$ in $Q$,
        \begin{equation*}\label{eq def bot up cong}
          \delta((q_1,\ldots,q_{n-1},p,q_{n+1},\ldots,q_m),f) \sim \delta((q_1,\ldots,q_{n-1},p',q_{n+1},\ldots,q_m),f).
        \end{equation*}
      \end{enumerate}
  Two Bottom-Up congruent states can be said \emph{interchangeable} in the litterature~\cite{AT90}.
  The \emph{quotient automaton} of $A$ w.r.t. $\sim$ is the automaton $A_\sim=(\Sigma,Q_\sim,F_\sim,\delta')$ with
  \begin{equation}\label{eq def quot}
    \delta'(([q_1],\ldots,[q_m]),f) = \{\phi(q)\mid q\in \delta((q_1,\ldots,q_m),f)\}.
  \end{equation}
  Since $A_\sim$ can be computed from a canonical morphism associated with $\phi$, and as a direct corollary of Proposition~\ref{prop morph pres lang},
  \begin{proposition}\label{prop bot up cong pres lang}
    Quotienting by a Bottom-Up congruence preserves the recognized language.
  \end{proposition}
  Let us now show how to obtain the Father automaton by quotienting the Position automaton w.r.t. the following congruence.
  \begin{definition}
    The \emph{Father congruence} associated with a linear expression $E$ over an alphabet $\Sigma$ is the congruence $\sim$ defined by
    \begin{equation*}
      p \sim p' \Leftrightarrow \mathrm{Father}(\$(E),p) = \mathrm{Father}(\$(E),p').
    \end{equation*}
  \end{definition}
  Equivalently, the Father congruence is the \emph{kernel} of the function sending a symbol $p$ to $\mathrm{Father}(\$(E),p)$.
  \begin{proposition}\label{prop father cong est cong}
    The Father congruence of a linear expression $E$ is a Bottom-Up congruence for the transition function of the Position automaton of $E$.
  \end{proposition}
  \begin{proof}
    First, let us notice that following Lemma~\ref{lem eq fin state dollar father}, two equivalent states have the same finality.
    Moreover, two states are equivalent if and only if they admit the same fathers.
    Consequently, from the construction of the Position automaton (Definition~\ref{def glush bu}), for any two states $p$ and $p'$, the two following conditions are equivalent:
      \begin{enumerate}
        \item $p\sim p'$,
        \item for any symbol $f$ in $\Sigma_m$, for any integer $n\leq m$, for any $m-1$ states $q_1$, $\ldots$, $q_{n-1}$, $q_{n+1}$, $\ldots$, $q_m$ in $Q$,
          \begin{equation*}
            \delta((q_1,\ldots,q_{n-1},p,q_{n+1},\ldots,q_m),f) = \delta((q_1,\ldots,q_{n-1},p',q_{n+1},\ldots,q_m),f).
          \end{equation*}
      \end{enumerate}
      Since $\sim$ is reflexive, it is a Bottom-Up congruence.
    \qed
  \end{proof}
  \begin{proposition}\label{prop father quot pos}
    The Father Automaton associated with a linear expression $E$ is isomorphic to the quotient of the Position automaton of $E$ w.r.t. the Father congruence.
  \end{proposition}
  \begin{proof}
    Let us set
    \begin{align*}
      {\mathcal{P}_E} &= (\_,\_,\_,\delta), &
      {\mathcal{P}_E}_\sim &= (\_,\_,\_,\delta_\sim)&
      \mathcal{F}_E &= (\_,\_,\_,\delta').
    \end{align*}
    Let us consider the functions $\phi$ and $\phi'$ defined as follows
    \begin{align*}
      \phi([f]) &= \mathrm{Father}(\$(E),f), &
      \phi'(\mathrm{Father}(\$(E),f)) &= [f].
    \end{align*}
    Notice that since $f\sim f' \Leftrightarrow \mathrm{Father}(\$(E),f) = \mathrm{Father}(\$(E),f')$, the functions are both well-defined.
    Moreover, they are trivially the inverse of each other.
    Furthermore, since
    \begin{equation*}
      f\sim f' \Rightarrow ((\$,1) \in \mathrm{Father}(\$(E),f) \Leftrightarrow (\$,1) \in \mathrm{Father}(\$(E),f')),
    \end{equation*}
    $\phi$ preserves the finality of the states.
    Finally,
    \begin{align*}
      (([f_1],\ldots,[f_n]),g,[g]) \in \delta_\sim
      & \Leftrightarrow ((f_1,\ldots,f_n),g,g) \in \delta & \textbf{(Equation~\eqref{eq def quot})}\\
      & \Leftrightarrow \forall i\leq n, (g,i) \in \mathrm{Father}(E,f_i)  & \textbf{(Definition~\ref{def glush bu})}\\
      & \Leftrightarrow ((\mathrm{Father}(\$(E),f_1),\ldots,    \mathrm{Father}(\$(E),f_n)),g,\mathrm{Father}(\$(E),g))\in\delta' & \textbf{(Definition~\ref{def father aut})}.
    \end{align*}
    Hence $\phi$ and $\phi'$ are two inverse automata morphisms between ${\mathcal{P}_E}_\sim$ and $\mathcal{F}_E$.
    \qed
  \end{proof}
  As a direct consequence of Proposition~\ref{prop bot up cong pres lang}, Proposition~\ref{prop father cong est cong} and Proposition~\ref{prop father quot pos},
  \begin{corollary}
    The Father Automaton associated with a linear expression $E$ recognized $L(E)$.
  \end{corollary}
  Applying the delinearization morphism $\mathrm{h}$ from $\mathcal{F}_{\overline{E}}$ produces the Father automaton of any expression $E$.
  Finally,
  \begin{corollary}
    The Father Automaton associated with an expression $E$ recognizes $L(E)$.
  \end{corollary}
  \begin{example}
    The Father automaton $(\mathrm{Pos}(\overline{E}),{\mathrm{Pos}(\overline{E})}_\sim,{\mathrm{Root}(\overline{E})}_\sim,\delta)$ of the expression $\overline{E}$ defined in Example~\ref{ex:calcul fonctions glushkov}
    is obtained by merging the states $f_1$ and $g_2$ of $\mathcal{P}_E$, \emph{i.e.}:
    \begin{gather*}
      \begin{aligned}
        \mathrm{Pos}(E) & = \{[a],[b],\{f_1,g_2\},[f_3],[g_4]\}, &
      \mathrm{Root}(\overline{E}) &= \{[a],[f_1]\},
      \end{aligned}\\
      \begin{aligned}
        \delta &= \{(a,[a]), (b,[b]),  (([a],[a]),f_1,[f_1]), (([a],[f_1]),f_1,[f_1]),  (([f_1],a),f_1,[f_1]),\\
        & \qquad  (([f_1],[f_1]),f_1,[f_1]), ([f_3],g_2,[g_2]),(([b],[g_4]),f_3,[f_3]), ([a],g_4,[g_4])\}.
      \end{aligned}
    \end{gather*}
  \end{example}

  \section{The Compressed Father Automaton}\label{sec:father-compressed}
  Finally, let us show that similarly to the Position automaton, the Father automaton can be compressed too.
  \begin{definition}\label{def fat comp}
    The \emph{compressed Father automaton} $\mathcal{CF}(E)$ of a linear expression $E$ is the automaton $(\Sigma,\mathrm{Pos}(E),\mathrm{Root}(E),\delta)$ defined by
    \begin{equation*}
      \delta = \{(Q_1,\ldots,Q_k,f,\{\mathrm{Father}(\$(E),f)\}) \mid Q_i = \{\mathrm{Father}(\$(E),g) \mid (f,i) \in \mathrm{Father}(\$(E),g)\} \}.
    \end{equation*}
  \end{definition}
  In order to show that $\mathcal{CF}(E)$ recognizes $L(E)$, we can apply the same method as for the Father automaton.
  First, due to Equation~\ref{eq:extdeltaEnsComp}, the definition of a Bottom-Up congruence for $A$ is exactly the same (Equation~\eqref{eq def bot up cong}).

  The \emph{quotient} of a compressed automaton $A=(\Sigma,Q,F,\delta)$ w.r.t. a Bottom-Up congruence $\sim$ is the compressed automaton $A_\sim=(\Sigma,Q_\sim, F_\sim, \delta')$ where
  \begin{equation*}
    \delta'((Q_1,\ldots,Q_m),f) = \{\phi(q)\mid q\in \delta((q_1,\ldots,q_m),f) \wedge  \forall i\leq m, [q_i] \in Q_i\}.
  \end{equation*}
  Similarly to Lemma~\ref{lem comp rest trans state} and Proposition~\ref{prop mem chemin compress glu}, it can be shown that
  \begin{proposition}
    The compressed Father automaton is a quotient of the compressed Position automaton w.r.t. the Father congruence.
  \end{proposition}
  As a direct corollary,
  \begin{corollary}
    The compressed Father automaton and the Father automaton of a linear expression $E$ recognize $L(E)$.
  \end{corollary}
  Applying the delinearization morphism $\mathrm{h}$ from $\mathcal{CF}_{\overline{E}}$ produce the compressed Father automaton of any expression $E$.
  Finally, according to Equation~\eqref{eq:lien morph lang compress},
  \begin{proposition}
    The compressed Father automaton and the Father automaton of an expression $E$ recognize $L(E)$.
  \end{proposition}
  \begin{example}
    Let us consider the expression $\overline{E}$ defined in Example~\ref{ex:calcul fonctions glushkov}.
    The compressed Father automaton of $\overline{E}$ is represented at Figure~\ref{linear_compressed_automata_father}.
  \end{example}
  \begin{figure}[H]
  \centerline{
  \begin{tikzpicture}[node distance=2.5cm,bend angle=30,transform shape,scale=1]
    \node[accepting ,state, rounded rectangle] (f1)  {$\{f_1,g_2\}$} ;
     \node[state, below right of=g_2, node distance = 3cm, rounded rectangle] (f_3) {$[f_3]$};
    \node[state, below of= f_3,node distance=1.7cm, rounded rectangle] (b) {$[b]$};
    \node[state,  below left of=f_3, rounded rectangle] (g_4) {$[g_4]$};
    \node[state,  below of=f1,node distance = 1cm] (cerc) {};
    \node[state,  left of=f_3,node distance = 1.3cm] (cerc1) {};
    \node[state, double,  below left of=f1,node distance = 2.5cm, rounded rectangle] (a) {$[a]$};
    \draw (a) ++(-1cm,0cm) node {$a$}  edge[->] (a);
    \draw (b) ++(0cm,-1cm) node {$b$}  edge[->] (b);
    \path[->]
     (f_3) edge[->,  bend right,above right] node {$g_2$} (f1)
     (a) edge[->,below left,bend right] node {$g_4$} (g_4)
     (cerc) edge[->, left] node {$f_1$} (f1)
     (cerc1) edge[->, above] node {$f_3$} (f_3);
    \draw (cerc) ++(0.5 cm,-0.5cm)edge node [above,pos=0.5] {} (cerc)  edge[dashed,bend left=10] node[left,above,pos=0.5]{}(a)edge[dashed,bend right=30] node[left,above,pos=0.5]{}(f1);
    \draw (cerc) ++(-0.5 cm,-0.5cm)edge node [above,pos=0.5] {} (cerc)  edge[dashed] node[left,above,pos=0.5]{}(a)
    edge[dashed,bend left=30] node[left,above,pos=0.5]{}(f1);
    \draw (cerc1) ++(-0.5 cm,-0.5cm)edge node [above,pos=0.5] {} (cerc1)  edge[dashed,bend right]node[left,above,pos=0.6]{} (g_4);
    \draw (cerc1) ++(0.5 cm,-0.5cm)edge node [above,pos=0.5] {} (cerc1)  edge[dashed,bend left=15]node[left,above,pos=0.6]{} (b);
    \end{tikzpicture}}
    \caption{The compressed Father automaton of the expression $(f_1(a,a)+g_2(b))^{*_a}\cdot_b f_3(g_4(a),b)$.}
    \label{linear_compressed_automata_father}
  \end{figure}
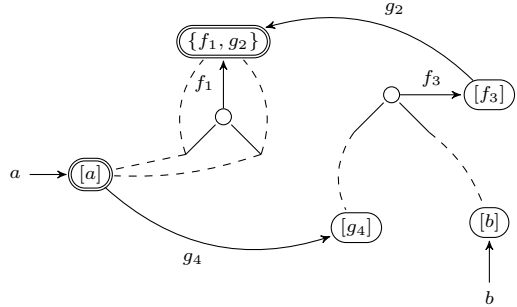

\section{Web Application}
  The computation of the position functions and the Glushkov and Father constructions have been implemented in a web application (made in Haskell, compiled into Java\-script using the \href{https://github.com/reflex-frp/reflex-platform}{\textsc{reflex platform}}, represented with \href{https://github.com/mdaines/viz.js}{\textsc{viz.js}}) in order to help the reader to manipulate the notions.
  From a regular expression, it computes the classical Top-Down Glushkov defined in~\cite{LOZ13}, and both the normal and the compressed versions of the Glushkov and Father Bottom-Up automata.

  This web application can be found \href{http://ludovicmignot.free.fr/programmes/fatherBotUp/index.html}{here}~\cite{AppWeb}.
  As an example, the expression $(f(a,a)+g(b))^{*_a}\cdot_bf(g(a),b)$ of Example~\ref{ex:calcul fonctions glushkov} can be defined from the literal input \texttt{(f(a,a)+g(b))*a.bf(g(a),b)}.

\section{Conclusion and Perspectives}

In this paper, we have shown how to compute the Bottom-Up Position and Father automata associated with a regular expression.
This construction is relatively similar to the classical ones defined over a word expression~\cite{Glu61,IY03}.
We have also proposed two reduced versions, the compressed Bottom-Up Position and Father automata, that can be easily defined for word expressions too.

Since this construction is related to the classical one, one can wonder if all the studies involving Glushkov word automata can be extended to tree ones (\cite{BDHLM04,BKW98,CZ00,Nic09}).
The classical Glushkov construction was also studied \emph{via} its morphic links with other well-known constructions.
The next step of our study is to extend Antimirov partial derivatives~\cite{Ant96} in a Bottom-Up way too (in a different way from~\cite{KM11}), using the Bottom-Up quotient defined in~\cite{CMOZ17}.

Moreover, we can also consider the transition compression in a Top-Down way, that is to also aggregate the destination states. This way, the Follow tree automaton~\cite{MOZ17} can also be compressed too.

Finally, this compression method can lead to heuristic study in order to choose a good aggregation to optimize any reduction.
To this aim, the decomposition techniques implemented to compute the Common Follow Set automaton~\cite{HSW01} can be considered.

\bibliography{biblio}

\begin{thebibliography}{10}

\bibitem{AT90}
Ad{\'a}mek, J., Trnkov{\'a}, V.:
\newblock Automata and Algebras in Categories.
\newblock Mathematics and its Applications. Springer Netherlands (1990)

\bibitem{Ant96}
Antimirov, V.M.:
\newblock Partial derivatives of regular expressions and finite automaton
  constructions.
\newblock Theor. Comput. Sci. \textbf{155}(2) (1996)  291--319

\bibitem{AMZ18}
Attou, S., Mignot, L., Ziadi, D.:
\newblock The bottom-up position tree automaton and its compact version.
\newblock In: accepted at {CIAA} 2018

\bibitem{BDHLM04}
Bouchou, B., Duarte, D., Alves, M.H.F., Laurent, D., Musicante, M.A.:
\newblock Schema evolution for {XML:} {A} consistency-preserving approach.
\newblock In: {MFCS}. Volume 3153 of Lecture Notes in Computer Science,
  Springer (2004)  876--888

\bibitem{BKW98}
Br{\"{u}}ggemann{-}Klein, A., Wood, D.:
\newblock One-unambiguous regular languages.
\newblock Inf. Comput. \textbf{140}(2) (1998)  229--253

\bibitem{CZ00}
Caron, P., Ziadi, D.:
\newblock Characterization of {G}lushkov automata.
\newblock Theor. Comput. Sci. \textbf{233}(1-2) (2000)  75--90

\bibitem{CMOZ17}
Champarnaud, J., Mignot, L., Sebti, N.O., Ziadi, D.:
\newblock Bottom-up quotients for tree languages.
\newblock Journal of Automata, Languages and Combinatorics \textbf{22}(4)
  (2017)  243--269

\bibitem{Glu61}
Glushkov, V.M.:
\newblock The abstract theory of automata.
\newblock Russian Mathematical Surveys \textbf{16} (1961)  1--53

\bibitem{HSW01}
Hromkovic, J., Seibert, S., Wilke, T.:
\newblock Translating regular expressions into small -free nondeterministic
  finite automata.
\newblock J. Comput. Syst. Sci. \textbf{62}(4) (2001)  565--588

\bibitem{IY03}
Ilie, L., Yu, S.:
\newblock Follow automata.
\newblock Inf. Comput. \textbf{186}(1) (2003)  140--162

\bibitem{KM11}
Kuske, D., Meinecke, I.:
\newblock Construction of tree automata from regular expressions.
\newblock {RAIRO} - Theor. Inf. and Applic. \textbf{45}(3) (2011)  347--370

\bibitem{LOZ13}
Laugerotte, {\'{E}}., Sebti, N.O., Ziadi, D.:
\newblock From regular tree expression to position tree automaton.
\newblock In: {LATA} 2013. (2013)  395--406

\bibitem{MY60}
McNaughton, R.F., Yamada, H.:
\newblock Regular expressions and state graphs for automata.
\newblock IEEE Transactions on Electronic Computers \textbf{9} (March 1960)
  39--57

\bibitem{AppWeb}
Mignot, L.:
\newblock Application: {G}lushkov tree automata.
\newblock \url{http://ludovicmignot.free.fr/programmes/fatherBotUp/index.html}
  Accessed: 2018-05-18.

\bibitem{MOZ17}
Mignot, L., Sebti, N.O., Ziadi, D.:
\newblock Tree automata constructions from regular expressions: a comparative
  study.
\newblock Fundam. Inform. \textbf{156}(1) (2017)  69--94

\bibitem{Nic09}
Nicaud, C.:
\newblock On the average size of {G}lushkov's automata.
\newblock In: {LATA}. Volume 5457 of Lecture Notes in Computer Science,
  Springer (2009)  626--637

\end{thebibliography}
\end{document}